\documentclass[submission,dvi]{eptcs}
\providecommand{\event}{LSFA 2025}%This is a template for producing LIPIcs articles. 
\usepackage{amsmath}
\usepackage{amsthm}
\usepackage[acronym]{glossaries}
\usepackage{todonotes}
\usepackage{tikz-cd}
\usepackage{cleveref}
\usepackage{quiver}
\usepackage{xstring}
\usepackage{bbold} % thanks Clem
\usepackage{bm}

\colorlet{Blue}{blue}
\colorlet{Orange}{orange}
\definecolor{ForestGreen}{rgb}{0.13, 0.55, 0.13}

\newtheorem{theorem}{Theorem}[section]

\newtheorem{lemma}[theorem]{Lemma}

\theoremstyle{definition}
\newtheorem{definition}{Definition}[section]

\usepackage{stmaryrd} %for Dana Scott brackets
\usepackage{newunicodechar}

\newunicodechar{≃}{\ensuremath{\simeq}}
\newunicodechar{λ}{\ensuremath{\lambda}}

\newcommand{\cont}[2]{#1 \mathbin{\triangleleft} #2}

\DeclareMathOperator{\Mon}{Mon}
\DeclareMathOperator{\Mnd}{Mnd}
\DeclareMathOperator{\Endo}{Endo}
\DeclareMathOperator{\Poly}{Poly}
\DeclareMathOperator{\Nat}{Nat}
\DeclareMathOperator{\Fin}{Fin}
\DeclareMathOperator{\Id}{Id}
\DeclareMathOperator{\ICMS}{ICMS}

\DeclareMathOperator{\id}{id}

\DeclareMathOperator{\St}{St}
\DeclareMathOperator{\refl}{refl}
\DeclareMathOperator{\Wr}{Wr}

\newcommand{\letin}[2]{\mathsf{let\:} #1 \mathsf{\:in\:} #2}
\newcommand{\ISt}[1]{\St^I_{#1}}
\newcommand{\catname}[1]{\mathsf{#1}}
\newcommand{\Set}{\catname{Set}}
\newcommand{\ICMSCat}{\catname{ICMSCat}} % There must be a better way
\newcommand{\Cat}{\catname{Cat}}
\newcommand{\SetI}{\Set^I}
\newcommand{\IC}{\catname{IC}}
\newcommand{\ICI}{\IC_I}
\newcommand{\C}{\mathcal{C}}
\newcommand{\D}{\mathcal{D}}
\newcommand{\ext}[1]{\llbracket{#1}\rrbracket}
\newcommand{\x}{\times}
\newcommand{\ox}{\mathbin{\otimes}}
\newcommand{\I}{\mathsf{I}}
\newcommand{\tosim}{\xrightarrow{\sim}}
\newcommand{\toI}{\to^I}
\newcommand{\upa}{\mathbin\uparrow}
\newcommand{\ula}{\mathbin\nwarrow}
\newcommand{\ura}{\mathbin\nearrow}
\newcommand{\iupa}[1]{\mathbin{\uparrow_{#1}}}
\newcommand{\iula}[1]{\mathbin{\nwarrow_{#1}}}
\newcommand{\iura}[1]{\mathbin{\nearrow_{#1}}}
\newcommand{\Peidx}{P\mathsf{e}_{\equiv}}

\newcommand{\Pibub}{\bullet^P}
\newcommand{\eP}{\e^P}
\newcommand{\smoosh}[1]{\overline{#1}}
\newcommand{\sbar}{\overline{s}}

\newcommand{\inv}[1]{{#1}^{-1}}

\DeclareMathOperator{\lop}{\mathsf{left}}
\DeclareMathOperator{\rop}{\mathsf{right}}
\DeclareMathOperator{\inl}{\mathsf{inl}}
\DeclareMathOperator{\inr}{\mathsf{inr}}
\DeclareMathOperator{\fst}{\mathsf{fst}}
\DeclareMathOperator{\snd}{\mathsf{snd}}
\newcommand{\eqdef}{\mathbin{\overset{\mathsf{def}}{=}}}
\newcommand{\semi}{\mathbin{\pmb ;}}
\newcommand{\pair}[2]{(#1, #2)}
\newcommand{\triple}[3]{(#1 , #2 , #3)}
\newcommand{\Unit}{\mathbb 1}

\newcommand{\fcart}{\sigma}
\newcommand{\fvert}{\pi}

\newcommand{\Bool}{\mathsf{Bool}}
\newcommand{\true}{\mathsf{true}}
\newcommand{\false}{\mathsf{false}}

\DeclareMathOperator{\data}{data\ }
\DeclareMathOperator{\where}{\ where}
\DeclareMathOperator{\var}{var}
\DeclareMathOperator{\app}{app}
\DeclareMathOperator{\lam}{lam}

\newcommand{\mact}{\mathbin{\blacktriangleright}}
\newcommand{\IWr}{\Wr^{\mact}}

\newlength{\alphabetheight}
\newlength{\alphabetdepth}
\newcommand{\agdahen}{%
  \begingroup
  \raisebox{-\alphabetdepth}
  {\includegraphics[height=\alphabetheight+\alphabetdepth]{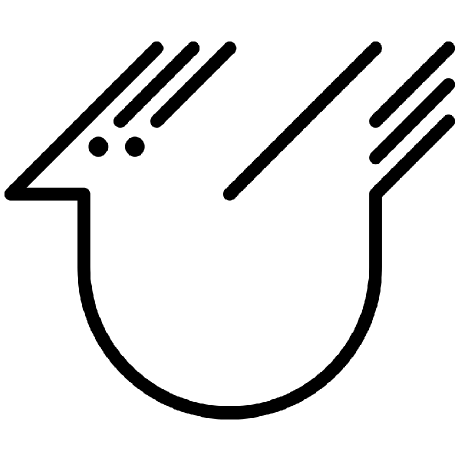}}%
  \endgroup
}

\newcommand{\agdaurl}{https://mikidep.github.io/indexed-monads/lsfa25}

\DeclareRobustCommand{\refagda}[1]{%
  \href{\agdaurl/#1}{\agdahen}%
}

\newcommand{\refagdaheading}[1]{%
  \texorpdfstring{\refagda{#1}}{}%
}

\newenvironment{agdaproof}[1]{%
  \begin{proof}[Proof \refagda{#1}]
}{%
\end{proof}
}

\DeclareSymbolFont{sfoperators}{OT1}{cmss}{m}{n}
\DeclareSymbolFontAlphabet{\mathsf}{sfoperators}
\makeatletter
\def\operator@font{\mathgroup\symsfoperators}
\makeatother

\newcommand{\ie}{\emph{i.e.\@}}
\newcommand{\eg}{\emph{e.g.\@}}
\newcommand{\Eg}{\emph{E.g.\@}}

\newacronym{itt}{ITT}{Intensional Type Theory}
\newacronym{mltt}{MLTT}{Martin-Löf Type Theory}
\newacronym{lccc}{LCCC}{Locally Cartesian Closed Category}
\newacronym{uip}{UIP}{Uniqueness of Identity Proofs}
\newacronym{hit}{HIT}{Higher Inductive Type}
\newacronym{hott}{HoTT}{Homotopy Type Theory}

\newcommand{\iprod}[1]{\prod_{\{#1\}}}

\newcommand{\iarg}[1]{_{#1}}

\newcommand{\comp}{\semi}

\newcommand{\idI}{\id^\I}
\newcommand{\compI}{\mathbin{\semi^I}}

\newcommand{\eps}{\varepsilon}

\newcommand{\mapF}{F}

\newcommand{\mapFi}{\mapF^{-1}}

\newcommand{\e}{\mathsf{e}}

\title{Monoid Structures on Indexed Containers}

\author{
Michele De Pascalis   
\institute{Tallinn University of Technology, Estonia}
\email{michde@taltech.ee}
\and
Tarmo Uustalu
\institute{Reykjavik University, Iceland}
\institute{Tallinn University of Technology, Estonia}
\email{tarmo@ru.is}
\and
Niccol{\`o} Veltr{\`i}
\institute{Tallinn University of Technology, Estonia}
\email{niccolo@cs.ioc.ee}
}
\def\titlerunning{Monoid Structures on Indexed Containers}
\def\authorrunning{M. De Pascalis, T. Uustalu \& N. Veltr{\`i}}

\begin{document}

\maketitle

\begin{abstract}
  Containers represent a wide class of type constructions that are relevant for
functional programming and (co)inductive reasoning. Indexed containers
generalize this notion to better fit the scope of dependently typed programming.
When interpreting types to be sets, a container describes an endofunctor on the
category of sets while an \(I\)-indexed container describes an endofunctor on
the category \(\SetI\) of \(I\)-indexed families of sets.

We consider the monoidal structure on the category of \(I\)-indexed containers
whose tensor product of containers describes the
composition of the respective induced endofunctors. We then give a combinatorial
characterization of monoids in this monoidal category, and we show how these
monoids correspond precisely to monads on the induced endofunctors on \(\SetI\).
Lastly, we conclude by presenting some examples of monads on \(\SetI\) that fall
under our characterization, including the product of two monads, indexed
variants of the state and the writer monads and an example of a free monad. The technical results of this work are accompanied
by a formalization in the proof assistant Cubical Agda.

\end{abstract}

\section{Introduction}%
\label{sec:Introduction}

%\subsection{Containers and polynomials}

Abbott et al.~\cite{Abbott2005} introduced the concept of (simple,
non-indexed) containers as a representation for a wide class of
parameterized datatypes, corresponding to strictly positive types, a
class of set endofunctors. The subsequent indexed generalization of
containers by Altenkirch et al.~\cite{Altenkirch2015} covers a class
of functors $\Set^I \to \Set^J$ between families of sets for $I$, $J$
some given sets of indices.

An example of datatype arising from an indexed container (in which
the index sets $I$ and $J$ are both equal to the set of natural numbers $\Nat$)
is the type of
vectors with bounded size.  In proof assistants such as Agda, the type
$\mathsf{Vec}\;A\;n$ of vectors with length $\le n$ can be implemented
as an inductive type:
\begin{align*}
  \textcolor{Orange}\data \: & \textcolor{Blue}{\mathsf{Vec}}\; (A :
  \textcolor{Blue}\Set) : \textcolor{Blue}\Nat \to
  \textcolor{Blue}\Set \: \textcolor{Orange}\where \\
          & \textcolor{ForestGreen}{\mathsf{nil}} : \{n : \textcolor{Blue}\Nat\} \to \textcolor{Blue}{\mathsf{Vec}}\;A\;n \\
  & \textcolor{ForestGreen}{\mathsf{cons}} : \{n : \textcolor{Blue}\Nat\} \to A \to \textcolor{Blue}{\mathsf{Vec}}\;A\;n \to \textcolor{Blue}{\mathsf{Vec}}\;A\;(1+n)
\end{align*}
But this can be equivalently defined as the type $\sum_{k \le n}
\Fin~k \to A$, whose elements are pairs of a number $k$ smaller than
$n$ (the actual length of the vector) and a function from $\Fin~k$
(the type of numbers $< k$) to $A$, i.e. a $k$-tuple of elements of $A$.
This type arises from a $\Nat$-indexed container with shapes at index
$n$ being the numbers $k \le n$ and positions in shape $k$ being
elements of $\Fin~k$.

Indexed containers are essentially a notational variant of the
(generalized) polynomials of Gambino and Hyland \cite{Gambino2004}, a
concept definable for any locally Cartesian closed category $\C$. A
polynomial from $I$ to $J$ interprets into a generalized polynomial
functor $\C/I \to \C/J$ between slice categories, where $\C := \Set$
is the archetypical choice of $\C$. Polynomials in this sense are a
variation on Joyal's \cite{Joyal1986} species, which interpret into
analytic functors.  Weber \cite{Web15} showed that the locally
Cartesian closed category can be replaced by a category with pullbacks.

The container and polynomial view of the same concept complement each other. The
polynomial view has advantages for theoretical developments, in particular
generalizations. For working with concrete examples or formalizing the theory
and examples in a type-theoretical proof assistant, the container view tends to
be smoother as argued e.g.\ by Finster et al.~\cite{Finster21}. 

% Another difference, stemming from the different applications considered, is that
% researchers working on containers are typically interested in general morphisms
% between containers~\cite{...} whereas the literature on polynomials focuses on
% Cartesian morphisms, i.e., those morphisms between polynomials whose
% interpretations into strong natural transformations are Cartesian
% \cite{Gambino2009}. (All polynomial functors are Cartesian.)

%% Michele: this is not really the case, cf. Finster21 sect. 2.2, 
%% which is arguably container literature but focuses on caresian morphisms.

Monads as endofunctors with a monoid structure are ubiquitous in both
type theory and category theory, with the type-theoretical (dependently
typed programming) uses most often coming from programming language
theory. Polynomial functors with a monad structure, called polynomial
monads, are an important special case and their in-depth study was
pioneered by Gambino and Kock~\cite{Gambino2009}.

Works on polynomials and containers also show variety in what they choose to
be the appropriate notion of polynomial/container morphism. Abbott et
al.~\cite{Abbott2005}, Altenkirch et
al.~\cite{Altenkirch2015} and Ahman et al.~\cite{ACU14,AHMAN2013} consider
general natural
transformations between the corresponding functors, while \eg{} Finster et
al.~\cite{Finster21} only target Cartesian strong natural transformations, as
noted in their Section 2.2. Gambino and Kock~\cite{Gambino2009} describe general natural transformations when defining polynomials, but only consider Cartesian
monads, \ie{} monads whose unit and multiplication are Cartesian natural
transformations.

Cartesian polynomial monads are of particular interest in the
meta-theory of type theory because they arise in the semantics of
type theory. They help organize Awodey's \cite{Awo18} natural models
of type theory, as demonstrated by Awodey and Newstead
\cite{AN18,New24} and Aberl\'e and Spivak \cite{AS24}.

Being a monad is not a mere property, there can be several choices of
unit and multiplication that make an endofunctor into a monad. Hence
it is natural to wonder whether, for container monads, these choices
can be beneficially analyzed in terms of their container
representation. Uustalu~\cite{Uustalu2017} answered this question for
non-indexed containers by unpacking what a monoid structure on a
container is by definition and transforming the result to a somewhat
simpler form.

\paragraph*{Our contribution}%
\label{sub:Scope of this article}

The contribution of this work is to extend the work of Uustalu~\cite{Uustalu2017} to indexed
containers and to provide examples of this characterization.

In detail, we outline the monoidal structure on indexed containers
that corresponds to composition of their extents. We then generalize
Uustalu's characterization to account for indexing by appropriately
adjusting the data and equations and prove that the resulting
structures correspond precisely to monad structures on indexed
container functors.

Lastly, we show how this characterization can be used practically, by
employing it to describe the product of two indexed container monads,
as well as indexed variants of the well-known state monad and writer
monads \cite{AU13} and, as an example of a free monad, the free monad which, when
applied to the family $\Fin$, gives the family of well-scoped
$\lambda$-terms.

We formalize this development in Cubical Agda.

\paragraph*{Related work}%
%\label{sub:Related Work}
Gambino and Kock~\cite{Gambino2009} studied the bicategorical
structure of (generalized) polynomials in arbitrary locally Cartesian closed categories: given an LCCC \(\C\) and objects $I, J$ of $\C$, their category \(\Poly_\C(I , J)\)
embeds fully and faithfully into the functor category \(\left[ \C / I , \C /
J \right]_\mathrm{strong}\) where the strengths are wrt.\ the canonical actions of $\C$ with its product monoidal structure on $\C/I$ and $\C/J$. Furthermore, one can define a bicategory whose objects are objects
in \(\C\) and whose \(\hom\)-categories are given by \(\hom(I , J) := \Poly_\C(I ,
J)\). Then, the previously described family of embeddings comes together into a
locally fully faithful 2-functor into the sub-2-category of \(\Cat\) having
\(\C\)-slices as 0-cells, polynomial functors as 1-cells, and strong natural
transformations as 2-cells.

In our work, we consider indexed containers where both indexings are
over the same fixed set \(I\). This restriction corresponds in
\cite{Gambino2009} to considering the full sub-bicategory of
\(\Poly_\C\) whose only 0-cell is \(I\), i.e., a monoidal category
which is analogous to the one studied in this paper.  It is also worth
noting that Gambino and Kock considered polynomial monads whose unit and
multiplication are Cartesian,
while our characterization is of general monads on extents of containers.

Proof assistant formalizations, mostly in Agda, of different chapters
of the theory of containers and polynomials have been carried out by
various groups of researchers. Ahman et al.~\cite{ACU14} formalized
directed containers (comonoids in containers). Finster et
al.~\cite{Finster21} formalized dependent polynomials/indexed
containers. Damato et al.~\cite{DAL25} formalized initial algebra and
final coalgebra constructions in containers. Joram and Veltri \cite{JV25}
developed a version of containers with symmetries, called action
containers.
Purdy and Damato \cite{PD25} have recently come up with a combinatorial characterization of distributive laws between monads arising from (non-indexed) containers.

\paragraph*{Structure of the paper}%
\label{sub:Paper structure}

In \Cref{sec:Indexed Containers} we review the category of indexed containers, that we call $\ICI$ (with $I$ being the indexing set), as well as its composition monoidal structure. In \Cref{sec:icmon} we recollect the notion of monoid internal to a monoidal category and provide a combinatorial
characterization of monoids in $\ICI$. \Cref{sec:Examples} collects a number of examples of interest in dependently-typed functional programming, all introduced using our combinatorial approach:
closure under Cartesian products, indexed variants of the state and the writer monads, and untyped $\lambda$-terms as an instance of a particular free monad.
We conclude in \Cref{sec:Conclusions}, where we discuss some ideas for future work.

\paragraph*{Formalization}%
\label{sub:Formalization}

This work has been formalized in the Cubical Agda proof assistant.
The choice of Cubical Agda is motivated by its support for functional
extensionality, which is employed in our development, crucially in the proof of fully-faithfulness of the interpretation in endofunctors and in proving the combinatorial monoid equations for the examples in \Cref{sec:Examples}.
Moreover, the use of Cubical Agda readily gives us the possibility of extending our work from sets to groupoids and higher-dimensional types.

\noindent The code is freely available at
\url{https://github.com/mikidep/indexed-monads/tree/lsfa2025}, or as clickable
HTML at \url{\agdaurl/Everything.html}. For the reader's
convenience, the formalized proofs and examples are decorated with a direct
hyperlink, in the form of the Agda logo (\agdahen).

\paragraph*{Notational conventions}%
\label{sub:Conventions}

We write composition of functions, and more generally morphisms in
categories, in diagrammatic order: the composition of morphisms \(A
\xrightarrow{f} B \xrightarrow{g} C\) is denoted by \(f \semi g\).
We use curly braces for implicit arguments of dependent functions.
When we want to make an implicit argument visible, we add it as a subscript.
For example, given $f : \iprod{x : X} \prod_{y : Y} Z ~x ~y$, we denote its application to the implicit argument $x : X$ and the explicit argument $y : Y$ by $f \iarg{x} ~y : Z ~x ~y$.
An underscore appearing in the left-hand side of a function definition represents an argument that is not used in the definition, and therefore does not appear in the right-hand side.
\(\Set\) denotes the usual category of sets and functions.
In Cubical Agda, a type $X$ is a set if, for all $x,y : X$, the equality type $x \equiv y$ is a proposition, i.e. for all $p,q : x \equiv y$, we have $p \equiv q$.
Given a set
\(I\), \(\SetI\) denotes the category of presheaves over \(I\) considered
as a discrete category. The objects of this category are families
of sets, i.e., functions of type \(I \to \Set\). The morphisms from
$X$ to $Y$ are functions of type $\iprod{i : I} X~i \to Y~i$. We write
$X \toI Y$ for the homset $\SetI(X, Y)$.
Identity in $\Set^I$ is denoted $\idI$ while composition is $\compI$.
That is, we define $(f \compI g)_i := f_i \comp g_i$ for
$f : X \toI Y$, $g : Y \toI Z$.
The singleton set is denoted \(\Unit\) with \(\star\) as unique inhabitant.

\section{Indexed containers}%
\label{sec:Indexed Containers}

We are interested in endofunctors on \(\SetI\) which are describable as the
extents of \emph{indexed containers} in the sense of Altenkirch et
al.~\cite{Altenkirch2015}. Note that in general the notion of indexed container
is parameterized in two indexing sets, however, in this work we only deal with
polynomials describing endofunctors, for which the two parameters have to
coincide. We shall recap the relevant definitions here.

\begin{definition}[Indexed container, cf.\ \textsf{ICont*} in \cite{Altenkirch2015}
    \refagda{IndexedContainer.html\#266}]
  An \emph{indexed container} (denoted \(\cont SP\)) is composed of an indexed
  set \(S : I \to \Set\), together with a family of indexed sets over it, \(P : \iprod{i : I} S~i \to I \to \Set\).
\end{definition}

\begin{definition}[Extent of an indexed container
  \refagda{IndexedContainer.html\#719}]
  Let \(\cont SP\) be an indexed container. We define its \emph{extent}
  \(\ext{\cont SP}\) as the endofunctor on \(\SetI\) defined on objects by
  \begin{equation*}
    \ext{\cont SP}\; X\; i := \sum_{s : S\; i} P \iarg{i}~ s~ \toI X
  \end{equation*}
and on maps by 
\begin{align*}
     & \ext{\cont SP}  : (X \toI Y) \to 
       (\ext{\cont SP}\; X \toI \ext{\cont SP}\; Y) \\
     & (\ext{\cont SP}\; f)_i~(s , v) := (s , v \compI f)
\end{align*}
\end{definition}

\begin{definition}[Morphism of indexed
  containers \refagda{IndexedContainer.html\#985}]
  \label{def:container-morphism}
  Let \(\cont SP\) and \(\cont {S'}{P'}\) be indexed containers. A morphism between them is a function of type $\iprod{i:I} \prod_{s : S\; i} \ext{\cont{S'}{P'}} (P \iarg{i}~ s)~ i$.
\end{definition}
  The type of container morphisms \(\ICI(\cont{S}{P}, \cont{S'}{P'})\) can be unfolded as follows:
  \begin{align*}
    & \iprod{i:I} \prod_{s : S\; i} \ext{\cont{S'}{P'}} (P \iarg{i}~ s)~ i
    \\
    & = \iprod{i:I} \prod_{s : S~ i}
    \sum_{s' : S'~i} P' \iarg{i}~ s' \toI P \iarg{i}~ s
    \\
    & \simeq \sum_{\sigma : S \toI S'} \iprod{i:I} \prod_{s : S~ i}
    P' \iarg{i}~ (\sigma~ i~ s) \toI P \iarg{i}~ s
  \end{align*}
  Given \(f : \ICI(\cont{S}{P}, \cont{S'}{P'})\), we denote as \(\fcart~ f\) and
  \(\fvert~ f\) the two components of the image of $f$ under the last equivalence above.

Indexed containers and their morphisms form a category $\ICI$ with identity
and composition defined as follows:
\begin{align*}
  & (\id_{\cont SP})_i~ s := \pair{s}{\id_{P_i s}}
  \\
  & (\alpha \semi \beta)_i~ s := \letin{
    \\
    & \qquad\pair{s'}{v'} \leftarrow \alpha_i~ s
    \\
    & \qquad\pair{s''}{v''} \leftarrow \beta_i~ s'
    \\
    & \quad
  }{\pair{s''}{v'' \compI v'}}
\end{align*}
The extent operation \(\ext{-}\) extends to a functor from \(\ICI\) to
\(\Endo(\SetI)\), the category with objects given by endofunctors on \(\SetI\)
and morphisms by natural transformations. The following lemma was already proved
in \cite{Altenkirch2015} for non-indexed containers.

\begin{lemma}
  \label{lem:ext-ff}
  \(\ext{-}\) is full and faithful.
\end{lemma}

\begin{agdaproof}{IndexedSetContainer.html\#3466}
  The proof employs (co)end calculus. For a reference, see \cite{Loregian2021}.
  Also, note that ends over discrete categories coincide with products.
  A more direct proof is carried out in the formalization.
  \begin{align*}
    {} & \Nat\left(\ext{\cont SP} , \ext{\cont{S'}{P'}}\right)
    \\
    \simeq & \int_{X} \SetI\left(\ext{\cont SP}\; X , \ext{\cont{S'}{P'}}\; X\right)
    \tag{nat.\ transfs.\ as an end}
    \\
    \eqdef & \int_{X} \prod_{i:I} \Set\left(\ext{\cont SP}\; X\; i ,
    \ext{\cont{S'}{P'}}\; X\; i\right)
    \\
    \eqdef & \int_{X} \prod_{i:I} \Set\left(\sum_{s : S\; i} \SetI(P\; s , X),
    \ext{\cont{S'}{P'}}\; X\; i\right)
    \\
    \simeq & \int_{X} \prod_{i:I} \prod_{s : S\; i} \Set\left(\SetI(P\; s , X),
    \ext{\cont{S'}{P'}}\; X\; i\right)
    \tag{dependent (un)currying}
    \\
    \simeq & \prod_{i:I} \prod_{s : S\; i} \int_{X} \Set\left(\SetI(P\; s , X),
    \ext{\cont{S'}{P'}}\; X\; i\right)
    \tag{Fubini rule for ends}
    \\
    \simeq & \prod_{i:I} \prod_{s : S\; i} \Nat\left( \SetI(P\; s , -),
    \ext{\cont{S'}{P'}}\; (-)\; i\right)
    \tag{nat.\ transfs.\ as an end}
    \\
    \simeq & \prod_{i:I} \prod_{s : S\; i} \ext{\cont{S'}{P'}} (P\; s)\; i
    \tag{Yoneda}
  \end{align*}
\end{agdaproof}

\subsection{Monoidal categories and strong monoidal functors}%
\label{sub:Monoidal Categories}
% refs stolen from arXiv:0908.3347 
Monoidal categories \cite{MacLane1978,Joyal1993} are ubiquitous in category theory at
large, and in particular in its applications in computer science. We quickly
recap the notions we employ, and in doing so specify the notation and
conventions we adopt in the rest of the work.

A \emph{monoidal category} $(\C, \I, \ox)$ consists of a category \(\C\),
together with an object $\I : \C$,  a bifunctor \(\ox : \C \x \C \to \C\) and natural isomorphisms typed
\begin{align*}
  \lambda_X &: \I \ox X \tosim X \\
  \rho_X &: X \ox \I \tosim X \\
  \alpha_{X,Y,Z} &: (X \ox Y) \ox Z \tosim X \ox (Y \ox Z)
\end{align*}

For these, the following two diagrams have to commute for any choice of \(X, Y,
Z, W\):
%% TRIANGLE EQUATION
% https://q.uiver.app/#q=WzAsMyxbMCwwLCIoWCBcXG94IFxcSSkgXFxveCBZIl0sWzIsMCwiWCBcXG94IChcXEkgXFxveCBZKSJdLFsxLDEsIlggXFxveCBZIl0sWzAsMSwiXFxhbHBoYV97WCwgXFxJLCBZfSJdLFsxLDIsIlxcaWRfWCBcXG94IFxcbGFtYmRhX1kiXSxbMCwyLCJcXHJob19YIFxcb3ggXFxpZF9ZIiwyXV0=&macro_url=https%3A%2F%2Fgist.github.com%2Fmikidep%2Fd1bc17f8807cdf32405738530cb7fe17%2Fraw%2Ffdac14f9a47d3ee969bd0e5082a7370bc479a84e%2Fmacros.tex
\[\tag{triangle eq.}\begin{tikzcd}
	{(X \ox \I) \ox Y} && {X \ox (\I \ox Y)} \\
	& {X \ox Y}
	\arrow["{\alpha_{X, \I, Y}}", from=1-1, to=1-3]
	\arrow["{\rho_X \ox \id_Y}"', from=1-1, to=2-2]
	\arrow["{\id_X \ox \lambda_Y}", from=1-3, to=2-2]
\end{tikzcd}\]

%% PENTAGON EQUATION
% https://q.uiver.app/#q=WzAsNSxbMCwxLCIoKFcgXFxveCBYKSBcXG94IFkpIFxcb3ggWiJdLFsxLDAsIihXIFxcb3ggWCkgXFxveCAoWSBcXG94IFopIl0sWzIsMSwiVyBcXG94IChYIFxcb3ggKFkgXFxveCBaKSkiXSxbMiwyLCJXIFxcb3ggKChYIFxcb3ggWSkgXFxveCBaKSJdLFswLDIsIihXIFxcb3ggKFggXFxveCBZKSkgXFxveCBaIl0sWzAsMSwiXFxhbHBoYV97VyBcXG94IFgsIFksIFp9Il0sWzEsMiwiXFxhbHBoYV97VywgWCwgWSBcXG94IFp9Il0sWzMsMiwiXFxpZF9XIFxcb3ggXFxhbHBoYV97WCwgWSwgWn0iLDJdLFswLDQsIlxcYWxwaGFfe1csIFgsIFl9IFxcb3ggXFxpZF9aIiwyXSxbNCwzLCJcXGFscGhhX3tXLCBYIFxcb3ggWSwgWn0iLDJdXQ==&macro_url=https%3A%2F%2Fgist.github.com%2Fmikidep%2Fd1bc17f8807cdf32405738530cb7fe17%2Fraw%2Ffdac14f9a47d3ee969bd0e5082a7370bc479a84e%2Fmacros.tex
\[\tag{pentagon eq.}\begin{tikzcd}[cramped,column sep=0em]
	& {(W \ox X) \ox (Y \ox Z)} \\
	{((W \ox X) \ox Y) \ox Z} && {W \ox (X \ox (Y \ox Z))} \\
	{(W \ox (X \ox Y)) \ox Z} && {W \ox ((X \ox Y) \ox Z)}
	\arrow["{\alpha_{W, X, Y \ox Z}}", from=1-2, to=2-3]
	\arrow["{\alpha_{W \ox X, Y, Z}}", from=2-1, to=1-2]
	\arrow["{\alpha_{W, X, Y} \ox \id_Z}"', from=2-1, to=3-1]
	\arrow["{\alpha_{W, X \ox Y, Z}}"', from=3-1, to=3-3]
	\arrow["{\id_W \ox \alpha_{X, Y, Z}}"', from=3-3, to=2-3]
\end{tikzcd}\]

%They are called respectively the triangle and pentagon equation.

Given two monoidal categories \((\C, \I, \ox)\) and \((\D, \I', \ox')\), a
\emph{(strong) monoidal functor} consists of a functor \(F : \C \to \D\),
equipped with 
\begin{itemize}
  \item an isomorphism \(\eps : \I' \tosim F~ \I\),
  \item a natural isomorphism \(\psi_{X, Y} : F~X \ox' F~Y \tosim F~(X \ox Y)\).
\end{itemize}

\subsection{The composition monoidal structure of indexed containers}%
\label{sub:mon_struct_ind_cont}

\(\ICI\) features a
composition monoidal category structure, which we shall describe here explicitly.

\begin{definition}[Identity container
  \refagda{IndexedContainer.MonoidalCategory.html\#186}]
  The \emph{identity container} $\I$ is defined as
  \(\cont{S}{P}\), where:
  \begin{align*}
    & S~ i := \Unit \\
    & P \iarg{i} ~{\star} ~j := i \equiv j\
  \end{align*} 
\end{definition}

\begin{definition}[Composite container
  \refagda{IndexedContainer.MonoidalCategory.html\#299}]
  Let \(\cont{S^0}{P^0}\) and \(\cont{S^1}{P^1}\) be indexed containers
  over \(I\). Their composition \((\cont{S^0}{P^0}) \ox (\cont{S^1}{P^1})\) is defined as \(\cont{S}{P}\) where
  \begin{align*}
    & S~i := \ext{\cont{S^0}{P^0}}~ S^1~i
    \\
    & P \iarg{i}~ (s , s')~ k := \sum_{j : I} \sum_{p' : P^0 \iarg{i}~ s~ j} P^1 \iarg{j}~ (s'~
    p')~ k
  \end{align*}
\end{definition}

\begin{lemma}
  \label{lem:ic-moncat}
  \((\ICI, \I, \ox)\) is a monoidal category.
\end{lemma}

\begin{agdaproof}{IndexedContainer.MonoidalCategory.html}
  Functoriality of \(\ox\) holds definitionally, while the unitors and associator
  (and their inverses) are defined by reassociating \(\Sigma\)-types and
  providing the necessary \(\refl\)'s when needed. All is carried out in detail
  in the formalization.
\end{agdaproof}

\begin{lemma}
  \label{lem:ext-strmon}
\(\ext{-}\) is a strong monoidal functor \((\ICI, \I, \ox) \to (\Endo(\SetI) ,
\Id , (\circ) )\).
\end{lemma}

\begin{agdaproof}{IndexedContainer.Properties.html\#StrongMonoidal}
  Similarly to the above, this is all reassociations and identity type
  bookkeeping, please refer to the formalization.
\end{agdaproof}

\section{Monoids, indexed containers and monads}%
\label{sec:icmon}

\subsection{Monoid objects}%
\label{sub:Monoid objects}

We recall here the definition of a monoid object in a monoidal category.

\begin{definition}[Monoid in a monoidal category]
  \label{def:monoid}
  Let \((\C, \I, \ox)\) be a monoidal category. A \emph{monoid} \((X , \eta
  , \mu)\) in \(\C\) is an object $X$ of $\C$ coming with two maps
  \(\eta : \I \to X\) 
  and \(\mu : X \ox X \to X\), making the following diagrams commute.
  % https://q.uiver.app/#q=WzAsNCxbMCwwLCJcXEkgXFxveCBYIl0sWzEsMCwiWCBcXG94IFgiXSxbMSwxLCJYIl0sWzIsMCwiWCBcXG94IFxcSSJdLFswLDEsImUgXFxveCBcXGlkIl0sWzEsMiwibSJdLFswLDIsIlxcbGFtYmRhIiwyXSxbMywxLCJcXGlkIFxcb3ggZSIsMl0sWzMsMiwiXFxyaG8iXV0=&macro_url=https%3A%2F%2Fgist.githubusercontent.com%2Fmikidep%2Fd1bc17f8807cdf32405738530cb7fe17%2Fraw
  \[\tag{unitality}\begin{tikzcd}
	{\I \ox X} & {X \ox X} & {X \ox \I} \\
	& X
	\arrow["{\eta \ox \id}", from=1-1, to=1-2]
	\arrow["\lambda_X"', from=1-1, to=2-2]
	\arrow["\mu", from=1-2, to=2-2]
	\arrow["{\id \ox \eta}"', from=1-3, to=1-2]
	\arrow["\rho_X", from=1-3, to=2-2]
\end{tikzcd}\]
% https://q.uiver.app/#q=WzAsNSxbMCwxLCJYIFxcb3ggWCJdLFsyLDAsIlggXFxveCAoWCBcXG94IFgpIl0sWzIsMSwiWCBcXG94IFgiXSxbMSwyLCJYIl0sWzAsMCwiKFggXFxveCBYKSBcXG94IFgiXSxbMSwyLCJcXGlkIFxcb3ggbSJdLFsyLDMsIm0iXSxbNCwxLCJcXGFscGhhX3tYLCBYLCBYfSJdLFs0LDAsIm0gXFxveCBcXGlkIiwyXSxbMCwzLCJtIiwyXV0=&macro_url=https%3A%2F%2Fgist.githubusercontent.com%2Fmikidep%2Fd1bc17f8807cdf32405738530cb7fe17%2Fraw
\[\tag{associativity}\begin{tikzcd}[column sep=small]
	{(X \ox X) \ox X} && {X \ox (X \ox X)} \\
	{X \ox X} && {X \ox X} \\
	& X
	\arrow["{\alpha_{X, X, X}}", from=1-1, to=1-3]
	\arrow["{\mu \ox \id}"', from=1-1, to=2-1]
	\arrow["{\id \ox \mu}", from=1-3, to=2-3]
	\arrow["\mu"', from=2-1, to=3-2]
	\arrow["\mu", from=2-3, to=3-2]
\end{tikzcd}\]
\end{definition}
%  We will refer to \(X\) as the \emph{carrier}
%  of the monoid, to \(\eta\) as its unit and to \(\mu\) as its multiplication.

\begin{definition}[Monoid morphisms]
  \label{def:cat-mon}
  Let \((X , \eta , \mu)\) and \((Y , \eta' , \mu')\) be monoids in \(\C\).
  A \emph{monoid morphism} between them is a morphism \(f : \C(X , Y)\),
  making the following diagrams commute:

% https://q.uiver.app/#q=WzAsNyxbMSwwLCJYIl0sWzEsMSwiWSJdLFswLDAsIlxcSSJdLFszLDAsIlggXFxveCBYIl0sWzMsMSwiWSBcXG94IFkiXSxbNCwwLCJYIl0sWzQsMSwiWSJdLFswLDEsImYiXSxbMiwwLCJcXGV0YSJdLFsyLDEsIlxcZXRhJyIsMl0sWzMsNCwiZiBcXG94IGYiLDJdLFszLDUsIlxcbXUiXSxbNSw2LCJmIl0sWzQsNiwiXFxtdSciLDJdXQ==&macro_url=https%3A%2F%2Fgist.githubusercontent.com%2Fmikidep%2Fd1bc17f8807cdf32405738530cb7fe17%2Fraw
\begin{equation}\label{eq:monoid-morph-laws}
\begin{tikzcd}
	\I & X && {X \ox X} & X \\
	& Y && {Y \ox Y} & Y
	\arrow["\eta", from=1-1, to=1-2]
	\arrow["{\eta'}"', from=1-1, to=2-2]
	\arrow["f", from=1-2, to=2-2]
	\arrow["\mu", from=1-4, to=1-5]
	\arrow["{f \ox f}"', from=1-4, to=2-4]
	\arrow["f", from=1-5, to=2-5]
	\arrow["{\mu'}"', from=2-4, to=2-5]
\end{tikzcd}
\end{equation}
\end{definition}

It can be shown that identities and compositions of monoidal morphisms are
themselves monoidal morphisms. Hence, monoids in a monoidal category \((\C, \I,
\ox)\) together with the above defined morphisms form in turn a category,
denoted \(\Mon(\C)\). The following lemma is relevant to the main result of this
work.

\begin{lemma}
  \label{lem:mff-creates-monoids}
  Let \((\C, I, \ox)\) and \((\D, \I', \ox')\) be monoidal categories, and let
  \(F : \C \to \D\) with \(\eps : I' \to F I\) and \(\psi_{X,Y} : F X \ox' F Y
  \to F (X \ox Y)\) be a full, faithful and strong monoidal functor. Then, \(F\)
  induces a full and faithful functor between \(\Mon(\C)\) and \(\Mon(\D)\), {\ie}
  \(\Mon(\C)\) is a full subcategory of \(\Mon(\D)\).
\end{lemma}

\begin{proof}
  Let \((X , \eta , \mu)\) be a monoid. It is well-known that $F$, being lax monoidal,
  induces a monoid structure on \(F X\) with:
  \begin{align*}
    \eta' &:= \I' \xrightarrow{\eps} F \I \xrightarrow{\mapF \eta} F X \\
    \mu' &:= F X \ox' F X \xrightarrow{\psi_{X,X}} F (X \ox X) \xrightarrow{\mapF \mu}
    F X
  \end{align*}

  \(F\) being strong monoidal, \(\eps\) and \(\psi\) are natural isomorphisms,
  therefore their precompositions \((\eps \comp {-})\) and \((\psi \comp {-})\) are in 
  turn isomorphisms, so we have:
  \begin{align*}
    {}     & \D(F X \ox' F X , F X) \\
    \simeq & \D(F (X \ox X) , F X)& \text{\((\psi_{X,X} \comp {-})\) is an iso} \\
    \simeq & \C(X \ox X , X) & \text{\(F\) is full and faithful}
  \end{align*}
  Notice that the right-to-left function underlying this isomorphism constructs the multiplication operation \(\mu' : \D(F X \ox' F X , F X)\) described above from a multiplication operation \(\mu : \C(X \ox X , X)\).
  
  Analogously, for \(\eps\) we have $\D(\I' , F X) \simeq \C(\I , X)$.

  It is only left to prove that under the inverses the monoid equations are
  still satisfied. Let \((F X ,  \eta' , \mu')\) be a monoid in \(\D\). Walking through
  the above isomorphism chain, this gets mapped to:
  \[
    (X , \mapFi(\eps^{-1} \comp  \eta') , \mapFi(\psi^{-1}_{X,X} \comp \mu'))
  \]

  This has to satisfy the monoid equations from \Cref{def:monoid} in \(\C\). We
  will only report here one of the unitality proofs, the rest of them can be
  worked out similarly. The following diagram has to commute.
  % https://q.uiver.app/#q=WzAsMyxbMCwwLCJcXEkgXFxveCBYIl0sWzEsMSwiWCJdLFsxLDAsIlggXFxveCBYIl0sWzAsMSwiXFxsYW1iZGEiLDJdLFsyLDEsIlxcbWFwaShcXGludlxcbXUgOyBtJykiXSxbMCwyLCJcXG1hcGkoXFxpbnZcXGV0YSA7IGUnKSBcXDsgXFxveCBcXDsgXFxpZCJdXQ==&macro_url=https%3A%2F%2Fgist.githubusercontent.com%2Fmikidep%2Fd1bc17f8807cdf32405738530cb7fe17%2Fraw%2F05c2ba3ce55646302f8b03365ed57eb4a269ef3e%2Fmacros.tex
  \[\begin{tikzcd}[column sep=7em]
	{\I \ox X} & {X \ox X} \\
	& X
	\arrow["{\mapFi(\inv\eps \comp \eta') \ox \id}", from=1-1, to=1-2]
	\arrow["\lambda_X"', from=1-1, to=2-2]
	\arrow["{\mapFi(\inv\psi \comp \mu')}", from=1-2, to=2-2]
\end{tikzcd}\]

Since \(F\) is faithful, it suffices to show that the image of the diagram
  under \(F\) commutes. That is achieved by the following pasting:
  % https://q.uiver.app/#q=WzAsNixbMCwwLCJGKEkgXFxveCBYKSJdLFswLDMsIkYgWCJdLFszLDAsIkYgKFggXFxveCBYKSJdLFszLDMsIkYgWCBcXG94JyBGIFgiXSxbMSwxLCJGIEkgXFxveCcgRiBYIl0sWzEsMiwiSScgXFxveCcgRiBYIl0sWzAsMSwiXFxtYXBcXGxhbWJkYSIsMl0sWzAsMiwiXFxtYXBfRihcXGludlxcbWFwX0YoXFxpbnZcXGV0YSBcXHNlbWkgZScpIFxcb3ggXFxpZF9YKSJdLFszLDEsIm0nIl0sWzIsMywiXFxpbnZcXG11Il0sWzAsNCwiXFxpbnZcXG11Il0sWzUsMSwiXFxsYW1iZGEnIiwyXSxbNCw1LCJcXGludlxcZXRhIFxcb3gnIFxcaWRfe0ZYfSJdLFs1LDMsImUnIFxcb3gnIFxcaWRfe0ZYfSJdLFs4LDUsIlxcdGV4dHsobW9ub2lkKX0iLDEseyJzaG9ydGVuIjp7InRhcmdldCI6MjB9LCJjb2xvdXIiOlsyNDAsNjAsNDBdLCJzdHlsZSI6eyJib2R5Ijp7Im5hbWUiOiJub25lIn0sImhlYWQiOnsibmFtZSI6Im5vbmUifX19LFsyNDAsNjAsNDAsMV1dLFsxMiw2LCJcXHRleHR7KFxcKEZcXCkgY29oZXJlbmNlKX0iLDEseyJzaG9ydGVuIjp7InNvdXJjZSI6MjAsInRhcmdldCI6MjB9LCJjb2xvdXIiOlsyNDAsNjAsNDBdLCJzdHlsZSI6eyJib2R5Ijp7Im5hbWUiOiJub25lIn0sImhlYWQiOnsibmFtZSI6Im5vbmUifX19LFsyNDAsNjAsNDAsMV1dLFs3LDksIlxcdGV4dHsoXFwoXFxpbnZcXG11XFwpIG5hdC4pfSIsMSx7InNob3J0ZW4iOnsidGFyZ2V0IjoyMH0sImNvbG91ciI6WzI0MCw2MCw0MF0sInN0eWxlIjp7ImJvZHkiOnsibmFtZSI6Im5vbmUifSwiaGVhZCI6eyJuYW1lIjoibm9uZSJ9fX0sWzI0MCw2MCw0MCwxXV1d&macro_url=https%3A%2F%2Fgist.githubusercontent.com%2Fmikidep%2Fd1bc17f8807cdf32405738530cb7fe17%2Fraw%2Fmacros.tex
\[\begin{tikzcd}
 	&&&& {F (X \ox X)} \\
	{F(I \ox X)} & {F I \ox' F X} & {I' \ox' F X} && {F X \ox' F X} \\
	&&&& {F X}
	\arrow["{\inv\psi_{X,X}}", from=1-5, to=2-5]
	\arrow[""{name=0, anchor=center, inner sep=0}, "{\mapF(\mapFi(\inv\eps \comp \eta') \; \ox \id)}", curve={height=-30pt}, from=2-1, to=1-5] 
	\arrow["{\inv\psi_{\I,X}}", from=2-1, to=2-2]
	\arrow[""{name=1, anchor=center, inner sep=0}, "{\mapF\lambda_X}"', curve={height=30pt}, from=2-1, to=3-5]
	\arrow["{\inv\eps \ox' \id}"', from=2-2, to=2-3]
	\arrow["{\eta' \ox' \id}"', from=2-3, to=2-5]
	\arrow[""{name=2, anchor=center, inner sep=0}, "{\lambda'_{FX}}"', curve={height=15pt}, from=2-3, to=3-5]
	\arrow["{\mu'}", from=2-5, to=3-5]
        \arrow["", "{\mapF(\mapFi(\inv\eps \comp \eta')) \ox \mapF \id}", curve={height=-15pt}, from=2-2, to=2-5]
	\arrow["{\text{(\(\inv\psi\) nat.)}}"{description, pos=0.7}, draw=none, from=2-3, to=0]
	\arrow["{\text{($F$ oplax mon.\ funct.)}}"{description}, draw=none, from=2-3, to=1]
	\arrow["{\text{($FX$ mon.)}}"{description, pos=0.6}, draw=none, from=2-5, to=2]
      \end{tikzcd}
    \]

    Lastly, it can be shown that for fixed \(M, M' : \Mon(\C)\), the (bijective)
    action of \(F\) on \(\Mon(\C)(M , M')\) respects and reflects the
    commutative diagrams in (\ref{eq:monoid-morph-laws}).
\end{proof}

\subsection{Monoid structures on indexed containers}%
\label{sub:Synthetic monoid structures on Indexed Containers}

In the specific monoidal category of indexed containers described in
Subsection~\ref{sub:mon_struct_ind_cont}, monoids can be presented in a more
combinatorial fashion. 

\begin{definition}[\(\ICMS\)
  \refagda{IndexedMonad.html\#ICMS}]
  \label{def:icms}

Let \(\cont{S}{P}\) be an indexed container. An \emph{indexed container
  monoid structure ($\ICMS$)} on \(\cont{S}{P}\) consists of the following data:
\begin{align*}
  \e & : \iprod{i : I} S\; i \\
  \bullet & : \iprod{i : I} \prod_{s : S\;i} (P \iarg{i}~s \toI S) \to S~ i \\
  \Peidx & : \iprod{i : I} \iprod{j : I} P \iarg{i}~ \e \iarg{i} ~ j \to i \equiv j \\
  {\upa} & : \iprod{i : I} \iprod{s : S\, i} \iprod{s' : P \iarg{i}\, s \toI S} \iprod{j : I} P \iarg{i} ~(s \bullet s')~j \to I \\
  {\ula} & : \iprod{i : I} \iprod{s : S\, i} \iprod{s' : P \iarg{i}\, s \toI S} \iprod{j : I} \prod_{p : P \iarg{i} \, (s \bullet s')\, j} P \iarg{i}~ s~ (\upa p) \\
  {\ura} & : \iprod{i : I} \iprod{s : S\, i} \iprod{s' : P \iarg{i}\, s \toI S} \iprod{j : I} \prod_{p : P \iarg{i} \, (s \bullet s')\, j} P \iarg{\upa p}~ (s' \iarg{\upa p} ~ (\ula p))~ j
\end{align*}

Arguments \(i, s, s', j\) for \(\upa\), \(\ula\) and \(\ura\) are implicit as
they can be often inferred by the type of \(p\). However, on occasion, we might
specify them as subscripts, in their respective order.

For any index $i : I$, shape $s : S\; i$ and index $j : I$, we require
\begin{align*}
  s \bullet \iarg{i} \eP  &\equiv s \label{eqn:e-unit-l} \tag{\textsf{e-unit-l}} \\
  \iula{i, s, \eP, j}~ p &\equiv p \label{eqn:ul-unit-l} \tag{\textsf{\(\ula\)-unit-l}}
             & \forall p : P \iarg{i}~ (s \bullet \eP)~ j\\
\e_i \bullet \iarg{i} \sbar&\equiv s \label{eqn:e-unit-r} \tag{\textsf{e-unit-r}} \\
\iura{i, \e_i, \sbar, j}~ p &\equiv p \label{eqn:ur-unit-r} \tag{\textsf{\(\ura\)-unit-r}}
             & \forall p : P \iarg{i}~ (\e_i \bullet \sbar)~ j\\
\end{align*}
where
\begin{align*}
  \eP & : P \iarg{i}~s \toI S \\
  \eP \iarg{j} ~ \_ & := \e \iarg{j} \\
  \sbar & : P \iarg{i}~\e \iarg{i} \toI S \\
  \sbar \iarg{j} ~  \_ & := s
\end{align*}  

For any index $i : I$, families of shapes $s : S\; i , s' : P \iarg{i} ~s \toI S , 
  s'' : \iprod {k: I} \prod_{q : P \iarg{i}~ s ~ k} P \iarg{k} ~(s' \iarg{k} ~ q) \toI S$, index $j : I$ and position $p : P \iarg{i} ((s \bullet s') \bullet \smoosh{s''}) ~ j$, we require
\begin{align*}
(s \bullet \iarg{i} s') \bullet \iarg{i} \smoosh{s''} &\equiv s \bullet \iarg{i} (s' \Pibub s'') \label{eqn:bullet-assoc} \tag{\textsf{\(\bullet\)-assoc}} \\
\iupa{i, s \bullet \iarg{i} s', \smoosh{s''}, j}~ p &\equiv \iupa{\upa p, s'(\ula p), s''(\ula p), j}~ (\iura{i, s, s' \Pibub s'', j}~ p) \label{eqn:up-urup-assoc} \tag{\textsf{\(\upa\)-\(\ura\upa\)-assoc}} \\
\iupa{i, s, s', \upa p}~ (\iula{i, s \bullet \iarg{i} s', \smoosh{s''}, j}~ p) &\equiv \iupa{i, s,
s' \Pibub s'', j}~ p \label{eqn:ulup-up-assoc} \tag{\textsf{\(\ula\upa\)-\(\upa\)-assoc}} \\
\iula{i, s, s', \upa p}~ (\iula{i, s \bullet \iarg{i} s', \smoosh{s''}, j}~ p) &\equiv \iula{i, s, s' \Pibub s'', j}~ p \label{eqn:ulul-ul-assoc} \tag{\textsf{\(\ula\ula\)-\(\ula\)-assoc}} \\
\iura{i, s, s', \upa p}~ (\iula{i, s \bullet \iarg{i} s', \smoosh{s''}, j}~ p) &\equiv \iula{\upa p, s' (\ula p), s'' (\ula p), j}~ (\iura{i, s, s' \Pibub s'', j}~ p)
\label{eqn:ulur-urul-assoc}
\tag{\textsf{\(\ula\ura\)-\(\ura\ula\)-assoc}}
\\
\iura{i, s \bullet \iarg{i} s', \smoosh{s''}, j}~ p &\equiv \iura{\upa p, s'
(\ula p), s'' (\ula p), j}~ (\iura{i, s, s' \Pibub s'', j}~ p) \label{eqn:ur-urur-assoc} \tag{\textsf{\(\ura\)-\(\ura\ura\)-assoc}}
\end{align*}
where
\begin{align*}
  s' \Pibub s'' & : P \iarg{i} ~ s \toI S \\
  (s' \Pibub s'') \iarg{k} ~q  & := s' \iarg{k} ~ q \bullet \iarg{k} s'' \iarg{k} ~ q \\
  \smoosh{s''} & : P \iarg{i} ~(s \bullet \iarg{i} s')  \toI S \\
  \smoosh{s''} \iarg{\ell} ~ r & := (s'' \iarg{\upa r} ~ (\ula r)) \iarg{\ell} ~ (\ura r)
\end{align*}

We will refer to the type of such structures as \(\ICMS(\cont{S}{P})\).
\end{definition}

Comparing this with the characterization of containers with monad structures in
Uustalu's work \cite[Section 3]{Uustalu2017}, it is evident that every piece of
data and every equation outlined there reappears in \Cref{def:icms}, adapted to
carry the necessary indices. However, the following pieces of data are new to
the indexed case:

\begin{itemize}
 \item $\Peidx{}_{i,j}~p$ is an extra coherence required of positions $p$ on all shapes
   \(\e_i\), arising from unitality: their source index must be the same as their
   target index.
 \item \(\upa p\) fulfills a similar function as \(\ula p\): where the
   latter maps a position in the composite shape to a position in the outer shape, the former specifies this new position's 
   index.
 \item \eqref{eqn:up-urup-assoc} and \eqref{eqn:ulup-up-assoc} are extra
   coherences deriving from how associativity acts on indices.
\end{itemize}

It is also worth noting that, as it was in \cite{Uustalu2017}, many of the
equations are heterogeneous, {\ie} the types of the equated terms only match
because of other non-definitional equations. As an instance, in equation
\eqref{eqn:ulul-ul-assoc}, the left-hand side has type
\(P\; s\; (\upa (\ula p))\), while the right-hand side has type
\(P\; s\; (\upa p)\): the two types are equal since \(\upa (\ula p)\) and \(\upa p\) are equal by \eqref{eqn:ulup-up-assoc}.
Additionally, the same terms often are required to
have propositionally, but not definitionally equal types, in the equated
expressions, {\eg} in equation \eqref{eqn:ulur-urul-assoc}, \(p\) should simultaneously
have types \(P\; ((s \bullet s') \bullet \smoosh{s''})\; j\) and \(P\; (s \bullet
(s' \Pibub s''))\; j\), which are provably equal due to
\eqref{eqn:bullet-assoc}.

Such heterogeneous equalities can also be found in the non-indexed case \cite{Uustalu2017}.
However, some new
ones appear due to \(\Peidx\), which are less trivial. In particular, for any $i : I, s : S\; i$ and $j : J$ we have 
\begin{align*}
  \Peidx{}_{i,j} ~ (\ura p) & : \upa p \equiv j
             & \forall p : P \iarg{i}~ (s \bullet \eP)~ j\\
  \Peidx{}_{i,j} ~ (\ula p) & : i \equiv \upa p 
             & \forall p : P \iarg{i}~ (\e_i \bullet \sbar)~ j
\end{align*}
The equations \eqref{eqn:ul-unit-l} and \eqref{eqn:ur-unit-r} type-check thanks to these two derivable equality proofs.

As should be clear by now, \Cref{def:icms} is constructed \emph{ad-hoc} to capture
the combinatorial structure of monoids in $(\ICI, \I, \ox)$. In fact, the
following holds:

\begin{lemma}[Characterization of monoids]
  \label{lem:icmon}
  Let \(\cont{S}{P}\) be an indexed container. The type of monoids structures in \(\triple
  \ICI \ox \I\) on \(\cont{S}{P}\) is equivalent to \(\ICMS(\cont{S}{P})\).
\end{lemma}

\begin{agdaproof}{IndexedMonad.html\#ICMonoid\%E2\%89\%83ICMS}
  We will illustrate how the data of an \(\ICMS\) is translated into the
  corresponding monoid structure. To that end, assume we have an \(\ICMS(\cont
  SP)\). We shall use it to define two container morphisms,  \(\eta :
  \IC(\I , \cont SP)\), and \(\mu : \IC((\cont SP) \ox (\cont SP) , \cont
  SP)\):
  \begin{align*}
   \eta \iarg{i} ~{\star} & := (\e \iarg{i}, \Peidx{} \iarg{i}) \\
   \mu \iarg{i} ~(s , s') & := (s \bullet \iarg{i} s', \mu^P \iarg{i} (s, s')) 
  \end{align*}
  where
  \begin{align*}
   (\mu^P \iarg{i} (s, s')) \iarg{k} ~p & := (\upa p, \ula p, \ura p)    
 \end{align*}  
  Showing that they satisfy the monoid equations, and that this mapping is a
  bijection, is mechanical and is better left to the accompanying
  formalization. However, intuitively, data is just being rearranged between the
  two types, and there is not much leeway for lossy operations. As one can imagine,
  the equations in the definition of \(\ICMS\) are tailored so that \(\eta\) and
  \(\mu\) satisfy the monoid equations.
\end{agdaproof}

The Agda formalization of \Cref{lem:icmon} requires manipulating and reasoning about the heterogeneous equalities in the definition of \(\ICMS\), which is a non-trivial exercise in intensional type theory.
We direct the interested reader to our Agda formalization for all the details.

After working out the proof, we were glad to observe that the theorem is true also when the container \(\cont{S}{P}\) is valued in types with an arbitrary h-level, {\ie} when $S ~i$ is not necessarily a set (what is called an h-set in homotopy type theory), and the same for $P \iarg{i} ~s ~j$.
The assumption of containers with $\Set$-valued shapes and positions is not instrumental in this case.
This suggests that the combinatorial specification of indexed monad containers in \Cref{def:icms}, as well as \Cref{lem:icmon}, can be adapted to higher dimensional settings, with sets replaced by higher groupoids, which is an avenue of research that we plan to investigate in the future (more on this in the conclusive section).

\begin{definition}[\(\ICMS\) morphism
  \refagdaheading{IndexedMonadMorphism.html}]
  \label{def:isICMM}

  Given two containers \((\cont S P)\) and \((\cont{S'}{P'})\) equipped with
  structures \((\e, \bullet, \Peidx,\upa,\ula,\ura) : \ICMS(\cont{S}{P})\) and
  \((\e', \bullet', \Peidx',\upa',\ula',\ura') : \ICMS(\cont{S'}{P'})\), and a
  container morphism \(f : \ICI(\cont S P , \cont{S'}{P'})\).
  We say that \(f\) is an \(\ICMS\) morphism between the two structures when it
  satisfies the following, for any index $i : I$, families of shapes $s : S\; i
  , s' : P \iarg{i} ~s \toI S$, index $j : I$ and position $p : P
  \iarg{i} ((\fcart~ f)\iarg{i}~ (s \bullet v)) ~ j$:

  \begin{align*}
    (\fcart~ f)\iarg{i} & \equiv \e'_i
  \label{eqn:f-e} \tag{\textsf{hom-\(\e\)}} \\
    \fvert~ f~ \e_i \compI \Peidx & \equiv \Peidx'
    \label{eqn:f-P-e-idx} \tag{\textsf{hom-\(\Peidx\)}} \\
    (\fcart~ f)\iarg{i}~ s \bullet' (\fvert~ f \compI v \compI \fcart~ f)
      & \equiv (\fcart~ f)\iarg{i}~ (s \bullet v)
      \label{eqn:f-bullet} \tag{\textsf{hom-\(\bullet\)}}
    \\
    \upa'~ p & \equiv \upa (\fvert~ f~ p)
    \label{eqn:f-up} \tag{\textsf{hom-\(\upa\)}}
    \\
    \fvert~ f(\ula'~ p) & \equiv \ula (\fvert~ f~ p)
    \label{eqn:f-ul} \tag{\textsf{hom-\(\ula\)}}
    \\
    \fvert~ f(\ura'~ p) & \equiv \ura (\fvert~ f~ p)
    \label{eqn:f-ur} \tag{\textsf{hom-\(\ura\)}}
  \end{align*}
\end{definition}

Such equations are obtained by applying the decomposition in
\Cref{def:container-morphism} to the diagrams in \Cref{def:cat-mon}, and in
fact:

\begin{lemma}[Characterization of monoid morphisms]
  \label{lem:isicmm}
  Let \((\cont SP,\eta,\mu)\) and \((\cont{S'}{P'},\eta',\mu')\) be objects of \(\Mon(\ICI)\).
  A morphism \(f : \ICI(\cont SP,\cont{S'}{P'})\) is a monoid morphism between \((\cont SP,\eta,\mu)\) and \((\cont{S'}{P'},\eta',\mu')\)
  if and only if it is an \(\ICMS\) morphism between \(E(\cont SP,\eta,\mu)\) and \(E(\cont{S'}{P'},\eta',\mu')\), where \(E\) is the equivalence in \Cref{lem:icmon}. 
\end{lemma}

\begin{agdaproof}{IndexedMonadMorphism.html\#4101}
  For a proof, we refer the reader to the formalization.
\end{agdaproof}

By virtue of this equivalence, the class of \(\ICMS\) morphisms is closed under
identities and composition, as is that of monoid morphisms wrt. \((\ICI, \I,
\ox)\). Hence, the following definition is well-posed.

\begin{definition}[The category \(\ICMSCat\)]
  We define \(\ICMSCat\) to be the category having indexed containers
  equipped with an \(\ICMS\) as objects, and compatible \(\ICMS\) morphisms as
  morphisms.
\end{definition}

\begin{theorem}[\(\ICMSCat \simeq \Mon(\ICI)\)]
  \label{thm:icmscat}
  The category of monoids in the monoidal category \((\ICI, \I, \ox)\) is
  equivalent to \(\ICMSCat\).
\end{theorem}

\begin{proof}
  It follows from the definition of \(\ICMSCat\), together with
  \Cref{lem:icmon} and \Cref{lem:isicmm}.
\end{proof}

We can now state the main result of this work:

\begin{theorem}
  \label{thm:main}
  The full subcategory of the category of \(\SetI\)-monads, whose endofunctors
  are extents of indexed containers, is equivalent to \(\ICMSCat\).
\end{theorem}

\begin{proof}
  \(\ext{-}\) is full, faithful (\Cref{lem:ext-ff}) and strong monoidal
  (\Cref{lem:ext-strmon}), so by \Cref{lem:mff-creates-monoids}, \(\Mon(\ICI)\)
  defines a full subcategory of \(\Mnd(\SetI) := \Mon (\Endo(\SetI))\). This full subcategory is equivalent to \(\ICMSCat\) by \Cref{thm:icmscat}.
\end{proof}

\section{Examples}%
\label{sec:Examples}

In this section we collect a few examples of monads arising from indexed containers,
as prescribed by Theorem \ref{thm:main}.
We show that pairs of indexed containers $\cont{S}{P}$ and combinatorial monoid structures in $\ICMS(\cont{S}{P})$ are closed under products.
We then describe indexed variants of state and writer monads.
We conclude with an example of a free monad that, when appropriately instantiated, produces a type of well-scoped untyped $\lambda$-terms.

\subsection{Product of two monads
\refagdaheading{IndexedMonad.Examples.Product.html}}%
\label{sub:exmpl-prod}

Given two monads \(\triple T \eta \mu\) and \(\triple {T'} {\eta'} {\mu'}\)
on \(\SetI\), one can endow the product endofunctor \(T \x T'\),
given by \((T \x T')\; X := T X \x T' X\), with a monad structure.
Furthermore, whenever \(T\) and \(T'\) are the extents of indexed containers,
their product is also the extent of an indexed container.
We use our characterization to describe the induced monad structure.

\begin{definition}[Product indexed container]
  Let \(\cont{S^0}{P^0}\) and \(\cont{S^1}{P^1}\) be indexed containers.
  Their product is defined as \(\cont{S}{P}\) where
  \begin{align*}
    & S~ i := S^0~ i \x S^1~ i
    \qquad \qquad
    P\iarg{i} ~(s,s') ~j := P^0 \iarg{i} ~s ~j + P^1 \iarg{i} ~s' ~j
  \end{align*}
\end{definition}

Suppose now that the two containers come with structure
\((\e^0, \bullet^0, \Peidx^0,\upa^0,\ula^0,\ura^0) :
\ICMS(\cont{S^0}{P^0})\) and
\((\e^1, \bullet^1, \Peidx^1,\upa^1,\ula^1,\ura^1) :
\ICMS(\cont{S^1}{P^1})\). We can define
\((\e, \bullet, \Peidx, \upa,\ula,\ura) : \ICMS(\cont{S}{P})\) as follows:
\begin{align*}
  \e \iarg{i} & := \pair {\e^0 \iarg{i}} {\e^1 \iarg{i}} \\
  (s_0,s_1) \bullet s' & := (s_0 \bullet^0 \lop s', s_1 \bullet^1 \rop s') \\
  \Peidx (\inl p) & := \Peidx^0 ~p \\
  \Peidx (\inr p) & := \Peidx^1 ~p \\
  \upa & := [ \upa^0 , \upa^1 ] \\
  \ula & :=  \ula^0 + \ula^1  \\
  \ura & :=  \ura^0 + \ura^1  
\end{align*}
Here $[-,-]$ is the copairing function out of a sum type, $\lop  : P \iarg{i} ~(s_0,s_1) \toI S$ is defined as
\((\lop s') \iarg{j} := \inl{} \semi s' \iarg{j} \semi \fst{}\), and \(\rop s'\) is similar with $\inr{}$ and $\snd{}$ in place of $\inl{}$ and $\fst{}$.
This definition satisfies the equations in \Cref{def:icms} directly by
\(\beta\)-equivalence. The explicit witnesses can be found in the accompanying
formalization.

\subsection{Indexed state monad
\refagdaheading{IndexedMonad.Examples.IndexedState.html}}%
\label{sub:exmpl-state}

Given a set \(E\), the endofunctor \(\St X := E \to E \x X\) has a canonical
monad instance:
\begin{equation*}
  \label[definition]{eqn:state-monad-instance}
\begin{aligned}
  \eta_X & : X \to (E \to E \x X) \\
  \eta_X & \; x\; e := \pair e x \\
  \mu_X & : (E \to E \x (E \to E \x X)) \to (E \to E \x X) \\
  \mu_X & \; f\; e := \letin{\pair {e'} {f'} = f\; e}{f'\; e'}
\end{aligned}
\end{equation*}
This is known to functional programmers as the \emph{state monad} with state
\(E\). As described in \cite{Uustalu2017}, one can observe that:
\begin{align*}
  {} & E \to E \x X \\
  {} \simeq & (E \to E) \x (E \to X) \\
  {} \simeq & \sum_{\_ : E \to E} E \to X
\end{align*}
which is precisely the extent of the non-indexed container with $E \to E$ as type of shapes and $E$ as the type of positions in each shape.

By considering a collection of sets of states \(E : \SetI\) instead of a single set $E : \Set$, we can define an indexed variant of the state monad.
Given a family of sets \(E : \SetI\), the \emph{indexed state endofunctor} is
$\ISt{E}\; X \; i := E ~i \to \sum_{j : I} E ~j \x X ~j$. 

Conceptually, we can describe an indexed stateful computation as follows: \(I\)
is a set of \emph{modes} that our stateful computer can be in. Not all states
are available in every mode, so the set of states is replaced by a collection of sets
\(E ~i\) containing the states available in mode \(i\).
From every mode \(i\) and state \(\e : E ~i\), a stateful
computation, besides returning the resulting value from \(X ~i\), points the
computer to a new mode \(j : I\), and sets a new state in \(E ~j\), in which
the next stateful computation can be performed.

By a similar reasoning as in the non-indexed case above, we can construct the following chain of equivalence.
\begin{align*}
  {} & E ~i \to \sum_{j : I} E ~j \x X ~j
  \\
  \simeq & \sum_{m : E \, i \to I} \prod_{e : E \,i} E ~{(m\; e)} \x X ~{(m\; e)}
  \tag{\(\Sigma\)-\(\Pi\) interchange}
  \\
  \simeq & \sum_{m : E \,i \to I} \left(\prod_{e : E \,i} E ~{(m\; e)}\right) \x
  \left(\prod_{e : E \,i} X ~{(m\; e)}\right)
  \tag{Universal property of product}
  \\
  \simeq & {\sum}_{\pair m \_ : \sum_{m : E \,i \to I} \left(\prod_{e : E \,i} E \,{(m\; e)}\right)}
  \left(\prod_{e : E \,i} X ~{(m\; e)}\right)
  \tag{associativity of \(\Sigma\)}
  \\
  \simeq & {\sum}_{s : E \,i \to \sum_{j : I} E \,j}
  \left(\prod_{e : E \,i} X ~{\fst(s\; e)}\right)
  \tag{\(\Sigma\)-\(\Pi\) interchange}
  \\
  \simeq & {\sum}_{s : E \,i \to \sum_{j : I} E \,j}
  \left(\prod_{j : I} \prod_{e : E \,i} \fst(s\; e) \equiv j \to X ~j\right)
  \tag{Yoneda}
  \\
  \simeq & {\sum}_{s : E \,i \to \sum_{j : I} E \,j}
  \left(\prod_{j : I} \left(\sum_{e : E \,i} \fst(s\; e) \equiv j\right) \to X ~j\right)
  \tag{currying}
\end{align*}
The latter type is the extent of the container \(\cont SP\), with:
\begin{align*}
  & S\; i := E ~i \to \sum_{j : I} E ~j
  \qquad \qquad
  P_i\; s\; j := \sum_{e : E \,i} \fst(s\; e) \equiv j
\end{align*}

A monoid structure in \(\ICMS(\cont SP)\) can be defined as follows:
\begin{align*}
  \e_i\; e & := (i, e)
  \\
  s \bullet s' & := \lambda e. ~\letin{(j,e') = s ~e}{s' \iarg{j} ~(e,\refl) ~e'}
  \\
  \Peidx & := \snd
  \\
  \iupa{i,s,\_,j} (e, \_) & := \fst(s\; e)
  \\
  \iula{i,s,\_,j} (e, \_) & := (e, \refl)
  \\
  \iura{i,s,\_,j} (e , p) & := (\snd(s\; e), p)
\end{align*}

When \(I = \Unit\), the extent of the container becomes isomorphic to the usual
state monad endofunctor, and the same is true for the monad structure.
The associativity-related equations in \Cref{def:icms} hold definitionally,
while the unit-related ones can be proved via functional extensionality by manipulating the equalities in the hypotheses. For further
details, please refer to the formalization.

\subsection{Indexed writer monad
\refagdaheading{IndexedMonad.Examples.Writer.html}}%
\label{sub:exmpl-writer}

Given a \(\Set\)-monoid \(\triple W \varepsilon {(\cdot)}\), the endofunctor \(\Wr ~X := W \x X\) can be given a monad
structure defined by:
\begin{equation*}
  \label[definition]{eqn:writer-monad-instance}
\begin{aligned}
  \eta_X & : X \to W \x X \\
  \eta_X & \; x := \pair \varepsilon x \\
  \mu_X & : W \x (W \x X) \to W \x X \\
  \mu_X & \; \pair w {\pair {w'} x} := \pair {w \cdot w'} x
\end{aligned}
\end{equation*}
It satisfies the monad equations by virtue of \(\triple W \varepsilon {(\cdot)}\)
being a monoid. This is usually referred to as the \emph{writer} monad, and is
clearly an instance of a container monad arising from the container \(\cont W
{(\lambda \_ . \Unit)}\).
One can immediately generalise the above to
the \(\SetI\)-endofunctor \(F ~X ~i := W \x X ~i\), essentially obtaining an
\(I\)-indexed product of writer monads.

However, we can go a step further,
allowing the endofunctor to change indices in a way that suits the monoid
structure on \(W\).
Recall that an \emph{action} of \(W\) on \(I\) is a monoid homomorphism from
\(\triple W \varepsilon {(\cdot)}\) to the monoid of endofunctions \(\triple {(I \to I)} {\id_I} {\semi}\).
Let \(\mact\) be such an action and define the endofunctor
\[
  \IWr X ~i := \sum_{w : W} X ~(w \mact i)
\]
Note that when we pick the trivial action \((w \mact i := i)\), we get back the
\(I\)-indexed product of writer monads described above.
The endofunctor \(\IWr\) is isomorphic to the extent of the container \(\cont
{S}{P}\) given by:
\begin{align*}
  & S~ i := W
  \qquad \qquad
  P \iarg{i} ~w~ j := w \mact i \equiv j
\end{align*}
We can then adapt the usual writer monad structure to this heterogeneously
indexed variant in the following way:
\begin{align*}
  \e \iarg{i} & := \varepsilon{}
  \\
  w \bullet w' & := w \cdot w'
  \\
  \Peidx & := \text{(\(\mact\) preserves \(\varepsilon\))}
         & \text{(witnesses \(\varepsilon \mact i \equiv j \to i \equiv j\))}
  \\
  \iupa{i,w,\_,j}\ \_ & := w \mact i
  \\
  \iula{i,w,\_,j}\ \_ & := \refl
  \\
  \iura{i,w,\_,j}\ \_ & := \text{(\(\mact\) preserves \((\cdot{})\))}
         & \text{(witnesses \((w \cdot w') \mact i \equiv j \to w' \mact (w
         \mact i) \equiv j\))}
\end{align*}
Notice that the purpose of the set of positions \(P_i ~w ~j\) is to express dependency constraints between indices $i$ and $j$ through the action \(\mact\).
Therefore, in this \(\ICMS\), the components on shapes are the same that appear in the non-indexed case, while the components on positions witness the fact that
the monad unit and multiplication respect the indexing constraints.

Lastly, the only data-relevant equations that need to be fulfilled are
\Cref{eqn:e-unit-l,,eqn:e-unit-r,,eqn:bullet-assoc}, which are granted directly
by the monoid structure on \(W\).

\subsection{Well-scoped $\lambda$-terms
\refagdaheading{IndexedMonad.Examples.ScopedLambda.html}}%
\label{sub:exmpl-sclam}

Given an indexed container \(\cont{S^0}{P^0}\), it is possible to construct a new container \(\cont{S}{P}\) and an element of \(\ICMS(\cont{S}{P})\) such that \(\ext{\cont{S}{P}}\) is the free monad on the endofunctor \(\ext{\cont{S^0}{P^0}}\).
We do not present in the paper the general construction for building the free indexed container monoid structure on an arbitrary indexed container \(\cont{S^0}{P^0}\).
We only show a particular example which illustrates the construction clearly.

In their work on indexed containers, Altenkirch et al. 
\cite[Introduction]{Altenkirch2015} point out that a family of well-scoped
$\lambda$-terms
can be defined in dependent type theory as the initial algebra for the
endofunctor on $\Set^{\Nat}$ given by \(X \mapsto \lambda n. ~\Fin ~n + (X ~n)^2 + X ~(n+1)\), where 
\(\Fin ~n := \{0, 1, \ldots, n-1\}\).

As reported by Gambino and Hyland \cite{Gambino2004}, given a polynomial endofunctor
\(F\) on a locally Cartesian closed category, the free monad \(F^*\) on \(F\) can be constructed as the
functor taking an object \(Y\) to the initial algebra of the endofunctor
\(X \mapsto Y + F\; X\), whenever such initial algebra exists.
The resulting endofunctor \(F^*\) is also polynomial and an appropriate monad structure on $F^*$ can be canonically built.

We can perform this construction for the endofunctor $F$ on $\Set^{\Nat}$ given by
\(F ~X ~n := (X ~n)^2 + X ~(n+1)\) and obtain \(F^* ~Y ~n := \mu X.\, Y ~n + (X ~n)^2 +
X ~(n+1)\) as the carrier of the free monad on \(F\).
Instantiating $Y$ to $\Fin$ gives a collection of types \(F^* \Fin\), which is the
same family of $\lambda$-terms described above.

Notice that the functor \(F\) is equivalent to the extent of the container \(\cont{S^0}{P^0}\), with \(S^0\;n := \Bool\) and
\begin{align*}
  & P^0 \iarg{n} ~\false ~m := \Bool \times (n \equiv m)
  \\
  & P^0 \iarg{n} ~\true ~m := (n + 1 \equiv m) 
\end{align*}
The functor $F^*$ can also be equivalently presented as the extent of an indexed container \(\ext{\cont SP}\).
Shapes can be given as an indexed W-type \(S\), that we describe as an inductive type family using Agda-like syntax.
\begin{align*}
  \textcolor{Orange}\data \: & S\; (n : \textcolor{Blue}\Nat) :
  \textcolor{Blue}\Set \: \textcolor{Orange}\where \\
& \textcolor{ForestGreen}{\var} : S\; n \\
& \textcolor{ForestGreen}{\app} : S\; n \to S\; n \to S\; n \\
& \textcolor{ForestGreen}{\lam} : S\; (n + 1) \to S\; n
\end{align*}

The family of positions $P$ is defined by recursion on the first shape in \(S ~n\) as follows:
\begin{align*}
  & P \iarg{n} ~\var ~m := (n \equiv m)
  \\
  & P \iarg{n} ~(\app ~M ~N) ~m := P \iarg{n} ~M ~m + P \iarg{n} ~N ~m
  \\
  & P \iarg{n} ~(\lam ~M) ~m := P \iarg{n+1} ~M ~m 
\end{align*}

One can then verify that \(\ext{\cont SP}\) satisfies the fixpoint equation
characterizing \(F^*\). The free monad structure is given as an element of 
\(\ICMS(\cont SP)\).
The operations $\bullet$, $\upa$, $\ula$ and $\ura$ are defined by structural recursion on the first
given shape in \(S ~n\). 
The required equations are proved by structural induction. 
\begingroup
\allowdisplaybreaks
\begin{align*}
  \e \iarg{n} & := \var
  \\
  \var{} \bullet \sigma & := \sigma ~\refl
  \\
  (\app M\; N) \bullet \sigma & := \app{} (M \bullet (\inl \semi \sigma)) (N \bullet (\inr \semi \sigma)) 
  \\
  (\lam M) \bullet \sigma & := \lam{} (M \bullet \sigma)  
  \\
  \Peidx ~p& := p
  \\
  \iupa{n,\var,\sigma,m} \_  & := n
  \\
  \iupa{n,\app M\; N,\sigma,m} (\inl p)  & :=  \iupa{n,M,(\inl \semi \sigma),m} p
  \\
  \iupa{n,\app M\; N,\sigma,m} (\inr p)  & :=  \iupa{n,N,(\inr \semi \sigma),m} p
  \\
  \iupa{n,\lam M,\sigma,m} p  & :=  \iupa{n+1,M,\sigma,m} p
  \\
  \iula{n,\var,\sigma,m} \_  & := \refl
  \\
  \iula{n,\app M\; N,\sigma,m} (\inl p)  & := \inl ( \iula{n,M,(\inl \semi \sigma),m} p)
  \\
  \iula{n,\app M\; N,\sigma,m} (\inr p)  & := \inr ( \iula{n,N,(\inr \semi \sigma),m} p)
  \\
  \iula{n,\lam M,\sigma,m} p  & :=  \iula{n+1,M,\sigma,m} p
  \\
  \iura{n,\var,\sigma,m} p  & := p
  \\
  \iura{n,\app M\; N,\sigma,m} (\inl p)  & :=  \iura{n,M,(\inl \semi \sigma),m} p
  \\
  \iura{n,\app M\; N,\sigma,m} (\inr p)  & :=  \iura{n,N,(\inr \semi \sigma),m} p
  \\
  \iura{n,\lam M,\sigma,m} p  & :=  \iura{n+1,M,\sigma,m} p
\end{align*}
\endgroup

We remark that, while $F^*$ is a monad on $\Set^{\Nat}$, the functor
$F^* \Fin : \Nat \to \Set $ is a monad relative to the functor
$\Fin : \Nat \to \Set$ and this latter relative monad structure is
more fundamental for well-scoped $\lambda$-terms that the former monad
structure. In particular, substitution of well-scoped $\lambda$-terms
is described by the multiplication of $F^* \Fin$, not the
multiplication of $F^*$.

\section{Conclusions}%
\label{sec:Conclusions}

In this paper we performed a combinatorial analysis of monad structures on  the extent of indexed containers, extending previous work by Uustalu on the non-indexed case.
The main technical contribution is a combinatorial characterization of monoid structures in the monoidal category of indexed containers, with the identity container and container composition as monoidal unit and tensor.
With this characterization, we generate many examples of interest in dependently-typed functional programming, such as indexed extensions of the state and the writer monad and a free monad producing the type of well-scoped $\lambda$-terms.
Our combinatorial analysis should serve as a tool for easing the construction of monad structures on certain endofunctors, for disproving that a candidate monad structure can possibly work, or for enumerative purposes.

Uustalu \cite{Uustalu2017} also discusses conditions ensuring that the resulting monads are Cartesian. We should be able to easily generalize these conditions to the indexed case.

When performing the formalization, we noticed that although our theoretical setting is based on \(\Set\)-families, many of the constructions presented in the paper work for types with arbitrary h-level, not just sets. In particular, there is no global assumption of Uniqueness of Identity Proofs.
The restriction to sets is necessary in the proof of \Cref{lem:ext-ff} (it is more generally needed to show that indexed containers form a category, where homtypes are required to be sets) and also employed in the formalization of examples.
In future work, we intend to extend our analysis to higher-categorical dimensions, moving from categories of (indexed) \emph{sets} to higher categories of (indexed) \emph{higher groupoids}.
The case of pseudofunctors and pseudomonads on 1-groupoids, a.k.a. types with h-level 1 in the sense of homotopy type theory \cite{UnivalentFoundationsProgram2013}, is already of interest.
Kock \cite{Kock2012} showed that polynomial functors on groupoids capture datatypes with symmetries, which are closely connected to analytic functors and Joyal's combinatorial species \cite{Joyal1986}.
\Eg\ a groupoid of finite multisets can be defined as the extent of a non-indexed container with shapes given by the ``groupoid of finite types'' and positions at a finite type given by its cardinality \cite{Finster21,Joram2023}.

A different avenue of future works concerns the use of indexed containers in the syntax and semantics of higher inductive types in HoTT.
Van der Weide and Geuvers \cite{Weide2019} introduce a schema for the specification of finitary set-truncated HITs, which are semantically interpreted as set-quotients of some initial algebras, and analogous studies exist for finitary 1-truncated HITs \cite{DybjerM18,VeltriW21}.
The expressivity of these schemata could be abundantly enhanced by moving to a notion of signature based on indexed containers.

\paragraph{Acknowledgments}
We are grateful to our reviewers for the numerous helpful comments they made. M.D.P.\ and N.V.\ were supported by the Estonian Research Council grant no.~PSG749. T.U.\ was supported by the Estonian Research Countil grant no.~PRG1210.

\bibliography{biblio}

\end{document}